\documentclass[11pt]{article}

\newtheorem{theorem}{Theorem}
\newtheorem{lemma}[theorem]{Lemma}

\newtheorem{claim}[theorem]{Claim}
\newtheorem{remark}[theorem]{Remark}
\newtheorem{definition}[theorem]{Definition}

\newenvironment{proof}
      {\medskip\noindent{\bf Proof :}\hspace{1em}}
      {\hfill$\Box$\medskip}

\newenvironment{algorithm}[1]
      {\begin{minipage}{0.965\textwidth}\medskip\noindent
       {\bf Algorithm \refstepcounter{theorem}\thetheorem\ #1}
       \begin{enumerate}\raggedright \sf }
      {\end{enumerate}\end{minipage}}

\usepackage{hyperref, nicefrac, multirow} 
\usepackage{comment,amsmath,amssymb,amsfonts} 
\usepackage{color,graphicx,url,subfigure}
\usepackage{url}

\newcommand{\nfrac}{\nicefrac}

\newcommand{\CH}{\mbox{${\mathcal H}$}}
\newcommand{\CB}{\mbox{${\mathcal B}$}}
\newcommand{\CG}{\mbox{${\mathcal G}$}}

\newcommand{\flag}{\mbox{\rm flag}}

\newcommand{\R}{\mathbb{R}}
\newcommand{\Rplus}{\R_+}
\newcommand{\vv}{\mbox{\boldmath $v$}}
\newcommand{\ww}{\mbox{\boldmath $w$}}
\newcommand{\yy}{\mbox{\boldmath $y$}}
\newcommand{\zz}{\mbox{\boldmath $z$}}
\newcommand{\ee}{\mbox{\boldmath $e$}}
\newcommand{\pp}{\mbox{\boldmath $p$}}

\newcommand{\xx}{\mbox{\boldmath $x$}}
\newcommand{\mmu}{\mbox{\boldmath $\mu$}}
\renewcommand{\aa}{\mbox{\boldmath $a$}}
\DeclareMathOperator*{\argmax}{arg\,max}
\DeclareMathOperator*{\argmin}{arg\,min}

\renewcommand{\AA}{\mbox{\boldmath $A$}}
\newcommand{\LL}{\mbox{\boldmath $L$}}

   \usepackage{algorithm}
   \usepackage{algorithmic}

   \newcommand{\bb}{\mbox{\boldmath $b$}}

   \newcommand{\Hlin}{\CH_{lin}}
   \newcommand{\Hsplc}{\CH_{splc}}
   \newcommand{\Hces}{\CH_{ces}}
   \newcommand{\Hleon}{\CH_{leon}}

   \newcommand{\bz}{\mbox{\boldmath $0$}}
   
   \newcommand{\bo}{\mbox{\boldmath $1$}}

   \newcommand{\iid}{\textsl{i.i.d.\ }} 
   \newcommand{\eg}{e.g.\ }

   \newcommand{\reals}{\mathbb{R}}
   \newcommand{\naturals}{\mathbb{N}}
   
   \newcommand{\e}{\mathrm{e}}
   \newcommand{\conv}{\mathrm{conv}}   
   \newcommand{\dist}{\mathrm{dist}}   
   \newcommand{\poly}{\mathrm{poly}}

   \newcommand{\err}{\mathrm{err}}

  \newcommand{\inner}[1]{\langle #1\rangle}

  \newcommand{\Acal}{{\mathcal A}}

  \newcommand{\Xcal}{{\mathcal X}}
  
  \newcommand{\Ycal}{{\mathcal Y}}

\date{}

\title{Learning Economic Parameters from Revealed Preferences}
\author{Maria-Florina Balcan$^1$
\and Amit Daniely$^2$
\and Ruta Mehta$^3$
\and Ruth Urner$^1$
\and Vijay V. Vazirani$^3$\\\\
\small{$^1$Department of Machine Learning, Carnegie Mellon University}\\
\texttt{\small ninamf@cs.cmu.edu, rurner@cs.cmu.edu}\\
\small{$^2$Department of Mathematics, The Hebrew University}\\
\texttt{\small amit.daniely@mail.huji.ac.il}\\
\small{$^3$School of Computer Science, Georgia Institute of Technology}\\
\texttt{\small rmehta@cc.gatech.edu, vazirani@cc.gatech.edu}
}

\begin{document}
\maketitle

\begin{abstract}
 A recent line of work, starting with Beigman and Vohra~\cite{BeigmanV06} and Zadimoghaddam and Roth~\cite{ZadimoghaddamR12}, has
 addressed the problem of {\em learning} a utility function from  revealed preference data. The goal here is to make use of past data
 describing the purchases of a utility maximizing agent when faced with certain prices and budget constraints in order to produce a
 hypothesis function that can accurately forecast the {\em future} behavior of the agent.

  In this work we advance this line of work by providing sample complexity guarantees and efficient algorithms for a number of
  important classes.  By drawing a connection to recent advances in  multi-class learning, we provide a computationally efficient
  algorithm with tight sample complexity guarantees ($\Theta(d/\epsilon)$ for the case of $d$ goods) for learning linear utility
  functions under a linear price model. This solves an open question in Zadimoghaddam and Roth~\cite{ZadimoghaddamR12}. Our technique
  yields numerous generalizations  including the ability to learn other well-studied classes of utility functions,  to deal with a
  misspecified model, and with non-linear prices.
  \medskip
  \medskip
  \medskip


\end{abstract}
\newpage

\section{Introduction}
A common assumption in Economics  is that  agents are utility maximizers, meaning that  the agent, facing
prices, will choose to buy the bundle of goods that she most prefers among all bundles that she can afford, according to some concave,
non-decreasing utility function \cite{MasColell}.
In the classical ``revealed preference'' analysis \cite{varian05}, the goal is to produce a model of the agent's utility function that can explain her behavior
based on past data.  Work on this topic has a long history in economics
\cite{uzawa,MC-RP1,MC-RP2,HSH,Rich,Afriat,var1,die,RPEmp1,RPEMP2}, beginning with the seminal work by Samuelson (1948) \cite{Samuelson}. Traditionally,
this work has focused on the ``rationalization'' or ``fitting the sample'' problem, in which explanatory utility functions are
constructively generated from finitely many agent price/purchase observations. For example, the seminal work of Afriat~\cite{Afriat} showed
(via an algorithmic construction) that any finite sequence of observations is rationalizable if and only if it is rationalizable by a
piecewise linear, monotone, concave utility function. Note, however, that just because a function agrees with a set of data does not
imply that it will necessarily predict future purchases well.

A recent exciting line of work, starting with Beigman and Vohra~\cite{BeigmanV06} introduced a statistical learning  analysis of the problem of learning the utility function from past data with the explicit formal goal of having predictive or forecasting properties.  The goal here is to make use of the observed data  describing the behavior of the agent (i.e., the bundles the agent bought when faced with certain prices and budget constraints) in order to produce a hypothesis function that can accurately predict or forecast the {\em future} purchases of a utility maximizing agent.
 \cite{BeigmanV06}  show that
without any other assumptions on utility besides monotonicity and
concavity, the sample complexity of learning (in a statistical or probably approximately
correct sense) a demand and hence utility function
is infinite. This shows the importance of focusing on important sub-classes since fitting just any monotone, concave function to the data will not be predictive for future events.

Motivated by this, Zadimoghaddam and Roth~\cite{ZadimoghaddamR12}  considered specific classes of utility functions including  the commonly used class of
linear utilities.  In this work, we advance this line of work by providing sample complexity guarantees and efficient
algorithms for a number of important classical classes (including linear, separable piecewise-linear concave (SPLC), CES and Leontief
\cite{MasColell}), significantly expanding the cases where we have strong learnability results.  At a technical level, our work
establishes connections between learning from revealed preferences and problems of multi-class learning, combining recent advances on
intrinsic sample complexity of multi-class learning based on compression schemes \cite{DanielyS14} with a new algorithmic analysis yielding time- and
sample-efficient procedures. We believe that this technique may apply to a variety of learning problems in economic and game theoretic
contexts.

\subsection{Our Results}
 For the case of linear utility functions, we establish a connection to  the so-called  structured prediction problem of \emph{$D$-dimensional linear classes} in
theoretical machine learning (see e.g.,~\cite{Collins00,Collins02,LaffertyMP01}). By using and improving very recent results of \cite{DanielyS14}, we provide a computationally efficient algorithm with  {\em tight} sample complexity guarantees for learning linear utility functions under a linear price model (i.e., additive over goods) for the statistical revealed
preference setting. This improves over the bound in Zadimoghaddam and Roth~\cite{ZadimoghaddamR12} by a factor of $d$ and resolves their open question concerning the right sample complexity of this problem.
In addition to noting that we can actually fit the types of problems stemming from revealed preference in the structured prediction framework of Daniely and Shalev-Shwartz \cite{DanielyS14}, we also provide a much more efficient and practical algorithm for this learning problem. We specifically show that we can reduce their compression based technique to a classic SVM problem which can be solved via convex programming\footnote{Such an algorithm has been used in the context of revealed preferences in a more applied work of \cite{Lahaie10}; but we prove correctness and tight sample complexity.}. This latter result could be of independent interest to Learning Theory.

The connection to the  structured prediction problem with $D$-dimensional linear classes is quite powerful and it  yields  numerous generalizations. It immediately implies strong sample complexity guarantees (though not
necessary efficient algorithms) for other important revealed preference settings. For linear utility functions we can deal with
non-linear prices (studied for example in \cite{Lahaie10}), as well as with a misspecified model --- in learning theoretic terms this
means the  agnostic setting where the target function is consistent with a linear utility function on a $1-\eta$ fraction of bundles;
furthermore, we can also  accommodate non-divisible goods. Other classes of utility functions including SPLC and CES can be readily
analyzed in this framework as well.

We additionally study  {\em exact} learning via revealed
preference queries: here  the goal of the learner is to determine the underlying utility function exactly, but it has more power since it can
 choose instances (i.e., prices and budgets) of its own choice and obtain the labels (i.e., the bundles the buyer  buys). We carefully exploit the structure of the optimal solution (which can be determined based on the KKT conditions) in order to design query efficient algorithms. This could be relevant for scenarios where sellers/manufacturers with many
different products have the ability to explicitly set desired prices  for exploratory purposes (e.g., with the goal to be able to predict how  demands
change with change in prices of different goods, so that they can price their
goods optimally).

As a point of comparison, for both statistical and the query setting, we also analyze learning classes of utility functions directly (from utility values instead of from revealed preferences). Table \ref{t:results_overview} summarizes our sample complexity bounds for learning from revealed preferences (RP) and from utility values (Value) as well as our corresponding bounds on the query complexity
(in the table we omit $\log$-factors).
Previously known results are indicated with a $*$.

\begin{table}[!h]
\label{t:results_overview}
\begin{tabular}{|c|c|c|c|c|}\hline
& {\bf RP, Statistical} & {\bf RP, Query} &  {\bf Value, Statistical} & {\bf Value, Query} \\ \hline\hline
Linear & $\Theta(d/\epsilon)$ & $O(nd)$ & $O(d/\epsilon)^*$  & $O(d)^*$ \\ \hline
SPLC (at most $\kappa$  & $\ O(\kappa d/\epsilon)$ (known \ & \multirow{2}{*}{$O(n\kappa d)$} &   \multirow{2}{*}{?}
&\multirow{2}{*}{$O(n\kappa d)$}  \\
segments per good) & \ segment lengths) & & & \\ \hline
\multirow{2}{*}{CES}  & $O(d/\epsilon)$ & \multirow{2}{*}{$O(1)$} &   $O(d/\epsilon)$  &\multirow{2}{*}{
$O(d)$}  \\
& (known $\rho$) & & (known $\rho$) & \\ \hline
Leontief & $O(1)$ & $O(1)$ & $O(d/\epsilon)$ & $O(d)$ \\
\hline
\end{tabular}
\caption{Markets with $d$ goods, and parameters of (bit-length) size $n$}
\end{table}

\section{Preliminaries}
\label{s:preliminaries}
Following the framework of \cite{ZadimoghaddamR12}, we consider a market that consists of 
 a set of agents (buyers), and
a set $\CG$ of $d$ goods of unit amount each. 
The prices of the goods are indicated by a price vector $\pp = (p_1, \ldots, p_d)$.
A buyer comes with a budget of money, say $B$, and intrinsic preferences over bundles of goods.
For most of the paper we focus on divisible goods.
A bundle of goods is represented by a vector $\xx = (x_1, \ldots, x_d) \in [0,1]^d$, where the $i$-th component $x_i$ denotes the amount of the $i$-th good in the bundle.
The price of a bundle is computed as the inner product $\inner{\pp, \xx}$.
Then the preference over bundles of an agent is defined by a non-decreasing, non-negative and concave utility function $U: [0,1]^d \rightarrow \Rplus$.
The buyer uses her budget to buy a bundle of goods that maximizes her utility.

In the revealed preference model, when the buyer
is provided with $(\pp,B)$, we observe the optimal bundle that she buys. 
Let this optimal bundle be denoted by $\CB_U(\pp,B)$, which is an optimal
solution of the following optimization problem: 

\begin{equation}\label{eq.ob}
\begin{array}{lcl}
\argmax_{\xx \in [0,\ 1]^d} & : & U(\xx) \\
s.t. && \inner{\pp,\xx} \le B \\ 
\end{array}
\end{equation}

We assume that if there are multiple optimal bundles, then the buyer will choose a cheapest one, i.e., let
$S=\argmax_{\xx \in [0,\ 1]^d} U(\xx)$ at $(\pp,B)$, then $\CB_U(\pp,B)\in \argmin_{\xx \in S} \inner{\xx, \pp}$.
Furthermore, if there are multiple optimal bundles of the same price, ties a broken according to some rule (\eg, the buyer prefers lexicographically earlier bundles).

\paragraph{Demand functions}
While a utility function $U$, by definition, maps bundles to values, it also defines a mapping from pairs $(\pp, B)$ of price vectors and budgets to an optimal bundles under $U$. 
We denote this function by $\widehat{U}$ and call it the \emph{demand function} corresponding to the utility function $U$.
That is, we have $\widehat{U}: \Rplus^d\times \Rplus \to [0,1]^d$, and $\widehat{U}(\pp, B) = \CB_U(\pp, B)$.
For a class of utility function $\CH$ we denote the corresponding class of demand functions by $\widehat{\CH}$.

\subsection{Classes of utility functions}
\label{ss:classes}
Next we discuss four different types of utility functions that we analyze in this paper, namely linear, SPLC, CES and Leontief
\cite{MasColell}, and define their corresponding classes formally.  Note that at given prices $\pp$ and budget $B$, $\CB_{U}(\pp,B) =
\CB_{\alpha U}(\pp,B)$, for all $\alpha>0$, i.e., positive scaling of utility function doesn't affect optimal bundles. Since we are
interested in learning $U$ by asking queries to $\CB_U$ we will make some normalizing assumptions in the following definitions. 
We start with the simplest and the most studied class of functions, namely linear utilities.

\begin{definition}
[Linear $\CH_{lin}$] A utility function $U$ is called linear if the utility from a bundle $\xx$ is linear in each good.
Formally, for some $\aa \in \Rplus^d$, we have $U(\xx) = U_{\aa}(x) = \sum_{j \in \CG} a_j x_j$. 
It is wlog to assume that $\sum_j
a_j=1$. We let $\Hlin$ denote the class of linear utility functions. 
\end{definition}

Next, is a generalization of linear functions that captures decreasing marginal utility, called separable piecewise-linear concave. 

\begin{definition}
[Separable Piecewise-Linear Concave (SPLC) $\CH_{splc}$] 
A utility function function $U$ is called SPLC if, $U(\xx)=\sum_{j \in \CG} U_j(x_j)$ where each $U_j:\Rplus \rightarrow \Rplus$ is non-decreasing piecewise-linear concave
function. The number of (pieces) segments in $U_j$ is denoted by $|U_j|$ and the 
$k^{th}$ segment of $U_j$  denoted by $(j, k)$. The slope of a segment specifies the rate at which the agent derives utility
per unit of additional good received. Suppose segment $(j, k)$ has domain $[a, b] \subseteq \Rplus$, and slope $c$. Then,
we define $a_{jk} = c$ and $l_{jk} = b-a$; $l_{j|U_j|}=\infty$). Since $U_j$ is concave, we have $a_{j(k-1)} > a_{jk},\ \forall k\ge
2$. We can view an SPLC  function, with $|U_j| \leq \kappa $ for all $j$, as defined by to matrices $\AA, \LL \in \Rplus^{d\times
\kappa}$ and we denote it by $U_{\AA\LL}$.  We let $\Hsplc$ denote the class of all SPLC functions.
\end{definition}

Linear and SPLC functions are applicable when goods are substitutes, i.e., one good can be replaced by another to
maintain a utility value. The other extreme is when goods are complementary, i.e., all goods are needed in some proportions
to obtain non-zero utility. Next, we describe a class of functions, used extensively in economic literature,
that captures
both substituteness and complementarity in different ranges.
\medskip

\begin{definition}
[Constant elasticity of substitution (CES) $\CH_{ces}$] A utility function $U$ is called CES if for some $-\infty < \rho \le 1$, and $\aa\in\Rplus^d$ we have $U(\xx)= U_{\aa\rho}(\xx)= (\sum_j a_j x_j^\rho)^{\nfrac{1}{\rho}}$.
Again it is wlog to assume that $\sum_j a_j=1$. Let $\CH_{ces}$ be the set of all CES functions. 
Further, for some fixed $\rho$, we let $\CH_{ces}^{\rho}$ denote the subclass of functions with parameter $\rho$. 
\end{definition}

Note that if $\rho=1$ for some CES function, then the function is linear, that is $\CH_{ces}^{1} = \CH_{lin}$.  Further, under CES functions with $\rho >0$, the goods behave
as substitutes. However, for $\rho \le 0$, they behave as complements, i.e., if an $x_j=0$ while $a_j>0$ the utility derived remains
zero, regardless of how much amounts of other goods are given. As $\rho\rightarrow -\infty$, we get Leontief function at the limit where
goods are completely complementary, i.e., a set of goods are needed in a specific proportion to derive any utility. 

\begin{definition}
[Leontief $\CH_{leon}$]
A utility function $U$ is called a Leontief function if $U(\xx)=\min_{j \in \CG} \nfrac{x_j}{a_j}$, where $\aa\ge 0$ and (wlog) $\sum_j a_j=1$.
Let $\CH_{leon}$ be the set of all Leontief functions on $d$ goods.
\end{definition}

In order to work with finite precision, in all the above definition we assume that the parameters defining the utility functions are
rational numbers of (bit-length) size at most $n$. 

\subsection{Learning models: Statistical \& Query}

We now introduce the formal models under which we analyze the learnability of utility functions.
We start by reviewing the general model from  statistical learning theory for multi-class classification.
We then explain the more specific model for learning from revealed preferences as introduced in \cite{ZadimoghaddamR12}.
Finally, we also consider a non-statistical model of exact learning from queries, which is explained last in this section.

\paragraph{\bf General model for statistical multi-class learning}
Let $\Xcal$ denote a \emph{domain}  set and let $\Ycal$ denote a \emph{label} set.
A \emph{hypothesis} (or \emph{label predictor} or \emph{classifier}), is a function $h:\Xcal\to \Ycal$, and a \emph{hypothesis class} $\CH$ is a set of hypotheses.
We assume that data is generated by some (unknown) probability distribution $P$ over $\Xcal$.
This data is labeled by some (unknown) \emph{labeling function} $l:\Xcal \to \Ycal$.
The quality of a hypothesis $h$ is measured by its \emph{error} with respect to $P$ and $l$:
\[
 \err_{P}^{l}(h) = \Pr_{x\sim P} [l(x) \neq h(x)],
\]
A \emph{learning algorithm} (or \emph{learner}) gets as input a sequence $S = ((x_1,y_1),\dots,$ $(x_m,y_m))$ and outputs a hypothesis. 

\begin{definition}[Multi-class learnability (realizable case)]\label{d:learnability}
 We say that an algorithm $\Acal$ \emph{learns} some hypothesis class $\CH\subseteq \Ycal^\Xcal$, if there exists a function $m: (0,1) \times (0,1) \to \naturals$  such that, for all distributions $P$ over $\Xcal$, and for all $\epsilon >0$ and  $\delta>0$, when given a sample $S = ((x_1,y_1),\dots,(x_m,y_m))$ of size at least $m = m(\epsilon,\delta)$ with the $x_i$ generated \iid from $P$ and $y_i = h(x)$ for some $h\in\CH$, then, with probability at least $1-\delta$ over the sample, $\Acal$ outputs a hypothesis $h_\Acal: \Xcal \to \Ycal$ with $\err_{P}^{h}(h_\Acal) \leq \epsilon$. 
\end{definition}

The complexity of a learning task is measured by its \emph{sample complexity}, that is, informally, the amount of data with which an optimal learner can achieve low error.
 We call the (point-wise) smallest function $m: (0,1) \times (0,1) \to \naturals$ that satisfies the condition of Definition \ref{d:learnability} the \emph{sample complexity of the algorithm $\Acal$ for learning $\CH$}. 
 We denote this function by $m[\Acal, \CH]$. 
 We call the smallest function $m: (0,1) \times (0,1) \to \naturals$ such that there exists a learning algorithm $\Acal$ with $m[\Acal,\CH]\leq m$ the \emph{sample complexity of learning $H$} and denote it by $m[\CH]$. 

\paragraph{\bf Statistical learning  from revealed preferences} As in \cite{ZadimoghaddamR12}, we consider a statistical learning setup where data is generated by a distribution $P$ over pairs of price vectors and budgets (that is, $P$ is a distribution over $\Rplus^d \times \Rplus$).
In this model, a learning algorithm $\Acal$ gets as input a sample $S=(((\pp_1,B_1), \CB_U(\pp_1, B_1)),$ $ \ldots, ((\pp_m,B_m), \CB_U (\pp_m, B_m)))$, where the $(\pp_i,B_i)$ are generated \iid from the distribution $P$ and are labeled by the optimal bundles under some utility function $U$. 
It outputs some function $\Acal(S): \Rplus^d \times \Rplus \to [0,1]^d$ that maps pairs of price vectors and budgets to bundles.
A learner is considered successful if it learns to predict a bundle of value that is the optimal bundles' value.
 
\begin{definition}[Learning from revealed preferences]
\label{d:learning_rev_pref}
An algorithm $\Acal$ is said to \emph{learn a class of utility functions $\CH$ from revealed preferences}, if for all $\epsilon, \delta >0$, there exists a sample size $m = m(\epsilon,\delta) \in \naturals$, such that, for any distribution $P$ over $\Rplus^d \times \Rplus$ (pairs of price vectors and budgets) and any target utility function $U\in \CH$, if $S=(((\pp_1,B_1), \CB_U(\pp_1, B_1)), \ldots, ((\pp_m,B_m), \CB_U (\pp_m, B_m)))$ is a sequence of \iid samples generated by $P$ with $U$, then, with probability at least $1-\delta$ over the sample $S$, the output utility function $\Acal(S)$ satisfies
\[
 \Pr_{(\pp,B)\sim P}\left[ U(\CB_U(\pp,B)) \neq U(\CB_{\Acal(S)}(\pp,B)) \right] \leq \epsilon.
\]
\end{definition}

Note that the above learning requirement is satisfied if the learner ``learns to output the correct optimal bundles''. 
That is, to learn a class $\CH$ of utility functions from revealed preferences, in the sense of Definition \ref{d:learning_rev_pref}, it suffices to learn the corresponding class of demand functions $\widehat{\CH}$ in the standard sense of Definition \ref{d:learnability} (with $\Xcal = \Rplus^d \times \Rplus$ and $\Ycal = [0,1]^d$).
This is what the algorithm in \cite{ZadimoghaddamR12} and our learning algorithms for this setting actually do. 
The notion of sample complexity in this setting can be defined analogously to the definition above.

\subsubsection{Model for exact learning from queries}

In the query learning model, the goal of the learner is to determine the underlying utility function exactly.
The learner can choose instances and obtain the labels of these instances from some oracle. 
A \emph{revealed preference query learning algorithm} has access to an oracle that, upon given the input (query) of a price vector and a budget $(\pp,B)$, outputs the corresponding optimal bundle $\CB_U(\pp,B)$ under some utility function $U$.
Slightly abusing notation, we also denote this oracle by $\CB_U$.

\begin{definition}[Learning from revealed preference queries]
 A learning algorithm \emph{learns a class $\CH$ from $m$ revealed preference queries}, if for any function $U\in\CH$, if the learning algorithm is given responses from oracle $\CB_U$, then after at most $m$ queries the algorithm outputs the function $U$. 
\end{definition}

Both in the statistical and the query setting, we analyze a revealed preference learning model as well as a model of learning classes of utility function ``directly'' from utility values. Due to limited space, these latter definition and results have been moved to the Appendix, Sections \ref{s:statistical_value} and \ref{ss:value_query}.

 \section{Efficiently learning linear multi-class hypothesis classes} 
 \label{sss:linear_classes}

 We start by showing that certain classes of multi-class predictors, so-called \emph{$D$-dimensional linear classes} (see Definition \ref{d:D_dim_linear_class} below), can be learnt efficiently both in terms of their sample complexity and in terms of computation.
 For this, we make use of a very recent upper bound on their sample complexity by Daniely and Shalev-Shwartz~\cite{DanielyS14}. 
 At a high level,  their result obtains strong  bounds on the sample complexity of $D$-dimensional linear classes (roughly $D/\epsilon$) by using an algorithm  and sample complexity analysis based on  a compression scheme ---  which roughly means that  the hypothesis produced can be uniquely described by a small subset of $D$ of the training examples.
 We show that their algorithm is actually equivalent to a multi class SVM formulation, and thereby obtain a computationally efficient algorithm with optimal sample complexity.
 In the next sections we then show how learning classes of utility functions from revealed preferences can be cast in this framework.
 
 \begin{definition}\label{d:D_dim_linear_class}
 A hypothesis classes $\CH\subseteq \Ycal^\Xcal$ is a \emph{$D$-dimensional linear class}, if there exists a function $\Psi: \Xcal \times \Ycal \to \reals^D$ such that for every $h\in \CH$, there exists a vector $\ww\in\reals^D$ such that
 $h(x) \in \argmax_{y\in \Ycal} \inner{\ww, \Psi(x,y)}$ for all $x\in\Xcal$. 
 We then also denote the class by $\CH_\Psi$ and its members by $h_{\ww}$.  
 \end{definition}
 
 For now, we assume that (the data generating distribution is so that) the set $\argmax_{y\in \Ycal} \inner{\ww, \Psi(x,y)}$ contains only one element, that is, there are no ties\footnote{The work of \cite{DanielyS14} handled ties using a ``don't know" label; to remove technicalities, we make this distributional assumption in this version of our work}.
The following version of the multi-class support vector machine (SVM) has been introduced by Crammer and Singer~\cite{CrammerS01}. 
\begin{algorithm}[h]
   \caption{Multi-class (hard) SVM \cite{CrammerS01}}
   \label{a:svm}
\begin{algorithmic}
   \STATE {\bfseries Input:} Sample $(\xx_1,y_1),\ldots,(\xx_m,y_m) \in \Xcal\times \Ycal$
   \STATE  Solve: $\ww = \argmin_{\ww\in\mathbb{R}^d} \|\ww\|$
   \STATE  such that $ \inner{\ww,\Psi(\xx_i,y_i)-\Psi(\xx_i,y)}\ge 1  \quad\forall  i\in[m], y \neq y_i$ 
   \STATE {\bfseries Return:} vector $\ww$ 
\end{algorithmic}
\end{algorithm}

\begin{remark}\label{r:svm_efficient}
Suppose that given $\ww\in\mathbb{R}^d$ and $x\in\Xcal$, it is possible to efficiently compute
some
$
y' \in \mathrm{argmax}_{y\not in \mathrm{argmax}_{y''} \inner{\ww,\Psi(x,{y''})}} 
\inner{\ww,\Psi(x,y)}.
$
That is, it is possible to compute a label $y$ 
in the set of ``second best'' labels.  In that case, it is not hard to see that SVM can be solved efficiently.  The reason is that this
gives  a separation oracle. SVM minimizes a convex objective subject to, possibly exponentially many, linear constraints. 
For a given $\ww$, 
a violated constraint can be efficiently detected (by one scan over the input sample) by observing that $\inner{\ww,
\Psi(x_i,y_i)-\Psi(x_i,y')} < 1$. 
\end{remark}

The following theorem on the sample complexity of the above SVM formulation, is based on the new analysis of linear classes by \cite{DanielyS14}.
We show that the two algorithms (the SVM and the one in \cite{DanielyS14}) are actually the same.

\begin{theorem}
\label{t:svm_for_d_linear}
Let $\CH_\Psi$ be some $D$-dimensional linear class.
Then the sample complexity of SVM for $\CH_\Psi$ satisfies 
 $
  m[SVM, \CH_\Psi](\epsilon,\delta) = O\left(\frac{D\log(1/\epsilon)+\log(1/\delta)}{\epsilon} \right).
 $
\end{theorem}

\begin{proof}
Let $S = (\xx_1,y_1),\ldots,(\xx_m,y_m)$ be a sample that is realized by $\CH_\Psi$.
That is, there exists a vector $\ww \in \reals^d$ with
$
 \inner{\ww, \Psi(\xx_i, y_i)} > \inner{\ww, \Psi(\xx_i,y)}$ for all  $y \neq y_i. 
$
Consider the set
$
Z=\{\Psi(x_i,y_i)-\Psi(x_i,y)\mid i\in [m],\;y\ne y_i\}.
$
The learning algorithm for $\CH_\Psi$ of \cite{DanielyS14}
outputs the minimal norm vector $\ww'\in \conv(Z)$. 
According to Theorem 5 in \cite{DanielyS14} this algorithm successfully learns $\CH_\Psi$ and has sample complexity $O\left(\frac{D\log(1/\epsilon)+\log(1/\delta)}{\epsilon} \right)$.
We will show that the hypothesis returned by that algorithm is the same hypothesis as the one returned by SVM.
Indeed, let $\ww$ be the vector that solves the SVM program and let $\ww'$ be the vector found by the algorithm of \cite{DanielyS14}. 
We will show that $\ww=\frac{\|\ww\|}{\|\ww'\|}\cdot \ww'$. This is enough since in that case $h_{\ww}=h_{\ww'}$.

We note that $\ww$ is the same vector that solves the {\em binary} SVM problem defined by the sample $\{(z,1)\}_{z\in Z}$.
It well known (see, \eg, \cite{MLbook}, Lemma 15.2) that the hyperplane defined by $\ww$ has maximal margin. 
That is, the unit vector $\ee=\frac{\ww}{\|\ww\|}$ maximizes the quantity
\[
\mathrm{mar}(\ee''):=\min\{\inner{\ee'',\zz}\mid \zz\in Z\}
\]
over all unit vectors $\ee''\in S^{d}$. 
The proof of the theorem now follows from the following claim:
\begin{claim}
Over all unit vectors, $\ee'=\frac{\ww'}{\|\ww'\|}$ maximizes the margin.
\end{claim}
\begin{proof}
Let $\ee''\ne \ee'$ be a unit vector. We must show that $\mathrm{margin}(\ee'')<\mathrm{margin}(\ee')$. 
Note that $\mathrm{margin}(\ee')>0$, since $\ww'$ is shown in \cite{DanielyS14} to realize the sample $S$ (that is $\inner{\ww, \zz} >0$ for all $\zz\in Z$ and thus also for all $\zz\in\conv{Z}$).
Therefore, we can assume w.l.o.g. that $\mathrm{margin}(\ee'')>0$. 
In particular, since $\mathrm{margin}(-\ee')=-\mathrm{margin}(\ee')<0$, we have that $\ee''\ne -\ee'$.

Since, $\mathrm{margin}(\ee'')>0$, we have that $\mathrm{margin}(\ee'')$ is the distance between the hyperplane $H''=\{\xx\mid \inner{\ee'',\xx}=0\}$ and $\mathrm{conv}(Z)$. Since  $\ee''\notin \{\ee',-\ee'\}$, there is a vector in $\vv\in H''$ with $\inner{\ee',\vv}\ne 0$. Now, consider the function
\[
t\mapsto \|t\cdot \vv-\ww'\|^2=t^2\cdot\|\vv\|^2+\|\ww'\|^2-2t\inner{\vv,\ww'}.
\]
Since the derivative of this function at $0$ is not $0$, for some value of $t$ we have
$
\dist(t\cdot \vv,\ww')<\dist(0,\ww').
$
Therefore,
$
\mathrm{margin}(\ee'')=\dist(H'',Z)\le \dist(t\cdot \vv,\ww')<\dist(0,\ww')=\mathrm{margin}(\ee').
$
\end{proof}

\end{proof}

\section{Statistical learning from revealed preferences}

In the next section, we show that learning utility functions from revealed preferences can in many cases be cast as learning a $D$-dimensional linear class $\CH_\Psi$ for a suitable encoding function $\Psi$ and $D$.
Throughout this section, we assume that the data generating distribution is so that there are no ties for the optimal bundle with respect to the agents' utility function (with probability $1$).
This is, for example, the case if the data-generating distribution has a density function.

\subsection{Linear}
Learnability of $\Hlin$ from revealed preferences is analyzed in \cite{ZadimoghaddamR12}. They obtain a bound of (roughly) $d^2/\epsilon$ on the sample complexity.
We show that the quadratic dependence on the number of goods is not needed. 
The sample complexity of this problem is (roughly) $d/\epsilon$. 

We will show that the corresponding class of demand functions $\widehat{\Hlin}$ is actually a $d$-dimensional linear class.
Since learnability of a class of utility functions in the revealed preference model (Definition \ref{d:learning_rev_pref}) is implied by learnability of the corresponding class of demand functions (in the sense of Definition \ref{d:learnability}), Theorem \ref{t:svm_for_d_linear} then implies the upper bound in the following result:

\begin{theorem}\label{t:linear_result}
 The class $\Hlin$ of linear utility functions is efficiently learnable in the revealed preference model with sample complexity   
  $
   O\left(\frac{d\log(1/\epsilon)+\log(1/\delta)}{\epsilon} \right).
  $
 Moreover, the sample complexity is lower bounded by
 $
 \Omega\left(\frac{(d-1) + \log(1/\delta)}{\epsilon}\right).
 $ 
\end{theorem}

\begin{proof}
Let $U_{\aa}$ be a linear utility function.
By definition, the optimal bundle given a price vector $\pp$ and a budget $B$ is  $\argmax_{\xx\in[0,1]^n, \inner{\pp, \xx} \le B} \inner{\aa, \xx}$.
Note that, for a linear utility function, there is always an optimal bundle $\xx$ where all (except at most one) of the $x_i$ are in $\{0,1\}$ (this was also observed in \cite{ZadimoghaddamR12}; see also Section \ref{ss:characterization_opt_bundles}).
Essentially, given a price vector $\pp$, in an optimal bundle, the goods are bought greedily in decreasing order of $a_i/p_i$ (value per price).

Thus, given a pair of price vector and budget $(\pp, B)$, we call a bundle $\xx$ \emph{admissible}, if $|\{i ~:~ x_i\notin\{0,1\}\}| \leq 1$ and $\inner{\pp, \xx} = B$. 
In case $\inner{\pp, \bo_d} = \sum_{i\in\CG} p_i \leq B$, we also call the all $1$-bundle $\bo_d$ admissible (and in this case, it is the only admissible bundle).
We now define the function $\Psi$ as follows:
\begin{equation*}
\Psi((\pp, B), \xx) = \left\{
\begin{array}{l}
 \xx  \quad\text{ if $\xx$ admissible} \\
 \bz_d \quad\text{ otherwise }
\end{array}
\right. 
\end{equation*}
where $\bz_d$ denotes the all-$0$ vector in $\reals^d$.
With this, we have $\CH_{\Psi} = \widehat{\CH_{lin}}$.

We defer the proof of the lower bound to the Appendix, Section \ref{s:lower_bound_for_Hlin}.    
To outline the argument, we prove that the Natarajan dimension of $\widehat{\CH_{lin}}$ is at least $d-1$ (Lemma \ref{l:Nidim}). 
This implies a lower bound for learning $\widehat{\CH_{lin}}$ (see Theorem \ref{nat_graph_upper_lower} in the Appendix).
It is not hard to see that the construction also implies a lower bound for learning $\Hlin$ in the revealed preference model.

To prove computational efficiency, according to Remark \ref{r:svm_efficient}, we need to show that for a linear utility function, we
can efficiently compute some 
\[
y' \in \mathrm{argmax}_{y'\notin \mathrm{argmax}_y \inner{\ww,\Psi(x,y)}} \inner{\ww,\Psi(x,y')};
\]
that is a second best bundle with respect to the mapping $\Psi$. This will be shown in Theorem \ref{thm:secondbest} of the next
subsection.  
\end{proof}

\subsubsection{Efficiently computing the second best bundle under linear utilities}
\label{s:computing_second_best}

It is known and easy to show (for example using KKT conditions for (\ref{eq.ob}), see Section \ref{ss:characterization_opt_bundles})
that an optimal bundle for the case of linear utility functions can be computed as follows: Sort the goods in decreasing order of
$\frac{a_j}{p_j}$, and keep buying in order until the budget runs out. The number of partially allocated goods in such a bundle is at
most one, namely the last one bought in the order. 

In this section show how to compute a second best {\em admissible bundle} (with respect to the mapping $\Psi$) efficiently. 
Recall that {\em admissible bundles} at prices $\pp$ and budget $B$ are defined (in the proof of the above theorem) to be the bundles
that cost exactly $B$ with at most one partially allocated good (or the all-$1$ bundle $\bo_d$, in case it is affordable). 
Note that, in case $\inner{\pp, \bo_d} \leq B$, \emph{any} other bundle is second best with respect $\Psi$. 
For the rest of this section, we assume that  $\inner{\pp, \bo_d} > B$.

At any given $(\pp,B)$ the optimal bundle is always admissible. We now design an $O(d)$-time algorithm to compute
the second best admissible bundle, i.e., $\yy \in \argmax_{\xx \mbox{ admissible }, \xx\neq \xx^*} \inner{\aa, \xx}$, where
$\xx^*$ is the optimal bundle.

At prices $\pp$, let $\frac{a_1}{p_1} \ge \frac{a_2}{p_2} \ge \dots \ge \frac{a_d}{p_d}$, and let the first $k$ goods be bought
at the optimal bundle, i.e., $k = \max_{j:\ x^*_j>0} j$. Then, clearly $\forall j<k,\ x^*_j=1$ and $\forall j>k,\ x^*_j =0$ as
$\xx^*$ is admissible. 

Note that, to obtain the second best admissible bundle $\yy$ from $\xx^*$, amounts of only first $k$ goods can be lowered
and amounts of only last $k$ to $d$ goods can be increased. Next we show that the number of goods whose amounts are lowered
and increased at exactly one each. In all the proofs we crucially use the fact that if $\frac{a_j}{p_j} > \frac{a_k}{p_k}$,
then transferring money from good $k$ to good $j$ gives a better bundle, i.e., $a_j \frac{m}{p_j} - a_k \frac{m}{p_k} >0$. 

\begin{lemma}\label{le:1}
There exists exactly one $j\ge k$, such that $y_j > x^*_j$.
\end{lemma}
\begin{proof}
To the contrary suppose there are more than one goods with $y_j > x^*_j$. 
Consider the last such good, let it be $l$. Clearly $l>k$, because the first good that can be increased is $k$. 
If $y_l < 1$ then there exists $j<l$ with $y_j=0$, else if $y_l = 1$ then there exists $j<l$ with $y_j<1$. In either case
transfer money from good $l$ to good $j$ such that the resulting bundle is admissible. Since,
$\frac{a_j}{p_j}>\frac{a_l}{p_l}$ it is a better bundle different from $\xx^*$. The latter holds because there is another
good whose amount still remains increased. A contradiction to $\yy$ being second best.
\end{proof}

\begin{lemma}\label{le:2}
There exists exactly one $j\le k$, such that $y_j < x^*_j$.
\end{lemma}
\begin{proof}
To the contrary suppose there are more than one goods with $y_j < x^*_j$. Let $l$ be the good with $y_l> x^*_l$; there is
exactly one such good due to Lemma \ref{le:1}. Let $i$ be the first good with $y_i < x^*_i$ and let $p$ be the good that is
partially allocated in $y$. If $p$ is undefined or $p\in\{i,l\}$, then transfer money from $l$ to $i$. to get a better
bundle. Otherwise, $p<l$ so transfer money from $l$ to $p$. In either case we can do the transfer so that resulting bundle is
admissible and is better than $\yy$ but different from $\xx^*$. A contradiction.
\end{proof}

Lemmas \ref{le:1} and \ref{le:2} gives an $O(d^2)$ algorithm to compute the second best admissible bundle, where we can
check all possible way of transferring money from a good in $\{1,\dots, k\}$ to a good in $\{k,\dots,d\}$. Next lemma
will help us reduce the running time to $O(d)$.

\begin{lemma}\label{le:3}
If $x^*_k<1$, and for $j>k$ we have $y_j>x^*_j$, then $y_k<x^*_k$. Further, if $x^*_k=1$ and $y_j>x^*_j$ then $j=k+1$.
\end{lemma}
\begin{proof}
To the contrary suppose, $y_k=x^*_k<1$ and for a unique $i<k$, $y_i<x^*_i$ (Lemma \ref{le:2}). 
Clearly, $y_i=0$ and $y_j=1$ because $0 < y_k < 1$. Thus, transferring money from $j$ to $k$ until either $y_j=0$ or $y_k=1$ gives a
better bundle different from $\xx^*$, a contradiction.

For the second part, note that there are no partially bought good in $\xx^*$ and $y_{k+1}=0$. To the contrary suppose $j>k+1$, then
transferring money from good $j$ to good $k+1$ until either $y_j=0$ or $y_{k+1}=1$ gives a better bundle other than $\xx^*$, a
contradiction.
\end{proof}

The algorithm to compute second best bundle has two cases. First is when $x^*_k<1$, then from Lemma \ref{le:3} it is clear that if an
amount of good in $\{k+1,\dots,d\}$ is increased then the money has to come from good $k$. This leaves exactly $d-1$ bundles to be
checked, namely when money is transferred from good $k$ to one of $\{k+1,\dots,d\}$, and when it is transferred from one of
$\{1,\dots,k-1\}$ good $k$. 

The second case is when $x^*_k=1$, then we only need to check $k$ bundles namely, when money is transferred from one of $\{1,\dots,k\}$
to good $k+1$. Thus, the next theorem follows.

\begin{theorem}\label{thm:secondbest}
Given prices $\pp$ and budget $B$, the second best bundle with respect to the mapping $\Psi$ for a utility function $U \in \Hlin$ at
$(\pp,B)$ can be computed in $O(d)$ time.
\end{theorem}

\subsection{Other classes of utility functions}

By designing appropriate mappings $\Psi$ as above, we also obtain bounds on the sample complexity of learning other classes of utility functions from revealed preferences.
In particular, we can employ the same technique for the class of SPLC functions with known segment lengths and the class of CES functions with known parameter $\rho$.
See Table \ref{t:results_overview} for an overview on the results and Section \ref{s:SPCL_CES_Leon} in the appendix for the technical details.

\subsection{Extensions}

Modeling learning tasks as learning a $D$-dimensional linear class is quite a general technique. 
We now discuss how it allows for a variety of interesting extensions to the results presented here.

\paragraph{Agnostic setting}
In this work, we mostly assume that the data was generated by an agent that has a utility function that is a member of some specific class (for example, the class of linear utilities).
However, this may not always be a realistic assumption. 
For example, an agent may sometimes behave irrationally and deviate from his actual preferences. 
In learning theory, such situations are modeled in the \emph{agnostic learning} model.
Here, we do not make any assumption about membership of the agents' utility function in some fixed class.
The goal then is, to output a function from some class, say the class of linear utility functions, that predicts the agents' behavior with error that is at most $\epsilon$ worse than the best linear function would.

Formally, the requirement on the output classifier $h$ in Definition \ref{d:learnability} then becomes $\err_{P}^{l}(h) \leq \eta + \epsilon$ (instead of $\err_{P}^{l}(h) \leq \epsilon$), where $\eta$ is the error of the best classifier in the class.
Since our sample complexity bounds are based on a compression scheme, and compression schemes also imply learnability in the agnostic learning model (see Section \ref{ss:compression_schemes}
  in the appendix), we get that the classes of utility functions with $D$-dimensional linear classes as demand functions that we have analyzed are also learnable in the agnostic model. That is, we can replace the assumption that the data was generated exactly according to a linear (or SPLC or CES) function with an assumption that the agent behaves according to such a function at least a $1-\eta$  fraction of the time.

\paragraph{Non-linear prices and indivisible goods}
So far, we looked at a market  where pricing is always linear and goods are divisible (see Section \ref{s:preliminaries}).
We note that the sample complexity results for $\Hlin, \Hsplc$, and $\Hces$  that we presented here actually apply in a broader context.
Prices per unit could vary with the amount of a good in a bundle (\eg \cite{Lahaie10}).
For example, there may be discounts for larger volumes.
Also, goods may not be arbitrarily divisible (\eg \cite{EcheniqueGW11}).  
Instead of one unit amount of each good in the market, there may then be a number of non-divisible items of each good on offer.
Note that we can still define the functions $\Psi$ to obtain a $D$-dimensional linear demand function class and the classes of utility functions discussed above are learnable with the same sample complexity (though not necessarily efficiently).

\paragraph{Learning preference orderings}
Consider the situation where we would like to not only learn the preferred choice (over a number $d$ of options) of an agent, but the complete ordering of preferences given some prices over the options.

We can model this seemingly more complex task as a learning problem as follows:
Let $\Xcal = \Rplus^d$ be our instance space of price vectors. 
Denote by $\Ycal = S_d$ the group of permutations on $d$ elements. 
Let a vector $\ww\in \mathbb{R}^d$ represent the unknown valuation of the agent, that is $w_i$ indicates how much the agent values option $i$.
Consider the functions $h_{\ww}:\Rplus^d\to S_d$ such that
$h_{\ww}(\pp)$ is the permutation corresponding to the ordering over the values $w_i/p_i$ (i.e. $\pi(1)$ is the index with the largest value per money $w_i/p_i$ and so on).

Finally, consider the hypothesis class
$
\CH_\pi=\{h_{\ww} ~:~ \ww\in \Rplus^d\}.
$
We show below hat $\CH_\pi$ is a $d$-dimensional linear class.
Therefore, this class can also be learned with sample complexity 
  $
   O\left(\frac{d\log(1/\epsilon)+\log(1/\delta)}{\epsilon} \right).
  $
With the same construction as for linear demand functions (see Lemma \ref{l:Nidim} in the appendix), we can also show that the Natarajan dimension of $\CH_\pi$ is lower bounded by $d-1$, which implies that this bound on the sample complexity is essentially optimal.

To see that $\CH_\pi$ is $d$-dimensional linear,  consider the map $\Psi:\Xcal\times S_d\to \mathbb R^d$ defined by
$
\Psi(\pp,\pi)=\sum_{1\le i<j\le d}\pi_{ij}\cdot((1/p_j) e_j - (1/p_i) e_i),
$
where, $\pi_{ij}$ is $1$ if $\pi(i)<\pi(j)$ and else $-1$; $e_1,\ldots,e_d$ is the standard basis of $\mathbb{R}^d$

\section{Learning via Revealed Preference Queries}

In this section we design algorithms to learn classes $\CH_{lin}$, $\CH_{splc}$, $\CH_{leon}$ or $\CH_{ces}$ using $\poly(n,d)$
revealed preference queries. (Recall that we have assumed all defining parameters of a function to be rationals of size (bit
length) at most $n$.)

\subsection{Characterization of optimal bundles}
\label{ss:characterization_opt_bundles}

In this section we characterize optimal bundles for linear, SPLC, CES and Leontief utility functions. In other words, given $(\pp,B)$
we characterize $\CB_U(\pp,B)$ when $U$ is in $\CH_{lin}$, $\CH_{splc}$, $\CH_{ces}$, or $\CH_{leon}$.  Since function $U$ is concave,
formulation (\ref{eq.ob}) is a convex formulation, and therefore Karush-Kuhn-Tucker (KKT) conditions characterize its
optimal solution \cite{BV,BSS}. For a general formulation $\min\{f(\xx)\ |\ g_i(\xx)\le 0,\ \forall i \le n\}$, the KKT conditions are
as follows, where $\mu_i$ is the dual variable for constraint $g_i(\xx)\le 0$. 

\[
\begin{array}{c}
L(\xx,\mmu) =  f(\xx)+ \sum_{i\le n} \mu_i g_i(\xx); \ \ \ \forall i\le n:\ \frac{dL}{dx_i} = 0 \\
\forall i\le n:\ \mu_ig_i(\xx)=0,\ \ \ g_i(\xx)\le 0,\ \ \ \mu_i\ge 0 
\end{array}
\]

In (\ref{eq.ob}), let $\mu$, $\mu_j$ and $\mu'_j$ be dual variables for constraints $\left<\pp,\xx\right>\le B$, $x_j\le 1$ and $-x_j \le 0$
respectively. Then its optimal solution $\xx^* =\CB_U(\pp,B)$ satisfies the KKT conditions: $\nfrac{dL}{dx_j}|_{\xx^*} =
-\nfrac{dU}{dx_j}|_{\xx^*} + \mu p_j +\mu_j -\mu'_j=0$, $\mu'_j x^*_j=0$, and
$\mu_j(x^*_j-1)=0$. Simplifying these gives us:

\begin{equation}\label{eq.kt}
\begin{array}{lcl}
\forall j \neq k,\ \ x^*_j>0,\ x^*_k=0\ \  & \Rightarrow \ \ & \frac{\nfrac{dU}{dx_j}|_{x^*}}{p_j} \ge \frac{\nfrac{dU}{dx_k}|_{x^*}}{p_k}\\
\forall j \neq k,\ \  x^*_j=1,\ 0 \le x^*_k<1 & \Rightarrow \  \ & \frac{\nfrac{dU}{dx_j}|_{x^*}}{p_j} \ge \frac{\nfrac{dU}{dx_k}|_{x^*}}{p_k}\\
\forall j \neq k,\ \  0 < x^*_j, x^*_k<1 & \Rightarrow\ \ & \frac{\nfrac{dU}{dx_j}|_{x^*}}{p_j}= \frac{\nfrac{dU}{dx_k}|_{x^*}}{p_k}
\end{array}
\end{equation}

\noindent{\bf Linear functions: } 
Given prices $\pp$, an agent derives $\nfrac{a_j}{p_j}$ utility per unit money spent on good $j$ (bang-per-buck). 
Thus, she prefers the goods where this ratio is maximum. Characterization of optimal bundle exactly reflects this,

\begin{equation}\label{eq.linob}
\begin{array}{lcl}
\forall j \neq k,\ \ x^*_j>0,\ x^*_k=0\ \  & \Rightarrow \ \ & \frac{a_j}{p_j} \ge \frac{a_k}{p_k}\\
\forall j \neq k,\ \  x^*_j=1,\ 0 \le x^*_k<1 & \Rightarrow \ \ & \frac{a_j}{p_j} \ge \frac{a_k}{p_k}\\
\forall j \neq k,\ \  0 < x^*_j, x^*_k<1 & \Rightarrow \ \ & \frac{a_j}{p_j}= \frac{a_k}{p_k}
\end{array}
\end{equation}

\medskip

\noindent{\bf SPLC functions:}
At prices $\pp$, the utility per unit money (bang-per-buck) on segment $(j,k)$ is $\nfrac{a_{jk}}{p_j}$. Clearly, the agent
prefers segments with higher bang-per-buck and therefore, if allowed, will buy segments in order of decreasing
bang-per-buck. Let $x^*_j$ in optimal bundle be 
be ending at $t^{th}$ segment. 
Then clearly segments 1 to $t-1$ are completely allocated, and segments $t+1$ to $|U_j|$ are not allocated at all. 
Accordingly define $\forall k<t,\ x^*_{jk} = l_{jk}$, $x^*_{jt}=x^*_j - \sum_{k<t} l_{jk}$, and $\forall k>t,\ x^*_{jk}=0$,
then similar to the conditions for linear function, these satisfy,

\begin{equation}\label{eq.splcob}
\begin{array}{lcl}
\forall (j,k) \neq (j',k'),\ \ x^*_{jk}>0,\ x^*_{j'k'}=0\ \  & \Rightarrow \ \ & \frac{a_{jk}}{p_j} \ge \frac{a_{j'k'}}{p_{j'}}\\
\forall (j,k) \neq (j'k'),\ \  x^*_{jk}=l_{jk},\ 0 \le x^*_{j'k'}<l_{jk} & \Rightarrow \ \ & \frac{a_{jk}}{p_j} \ge \frac{a_{j'k'}}{p_{j'}}\\
\forall (j,k) \neq (j'k'),\ \ 0 \le x^*_{jk},x^*_{j'k'}<l_{jk} & \Rightarrow \ \ & \frac{a_{jk}}{p_j} = \frac{a_{j'k'}}{p_{j'}}
\end{array}
\end{equation}
\medskip

\noindent{\bf CES utility functions:} Since $\frac{dU}{dx_j} = \frac{a_j U(x)^{1-\rho}}{x_j^{1-\rho}}$ and $-\infty < \rho < 1$, we
have $\lim_{x_j \to 0} \frac{dU}{dx_j} = \infty$. Therefore, conditions (\ref{eq.kt}) gives the following. 
$\forall j,\ x^*_j>0$, 

\begin{equation}\label{eq.cesob}
\begin{array}{llcl}
\forall j \neq k,& x^*_k < x^*_j=1 & \Rightarrow & \frac{a_j}{p_j} \ge \frac{a_k}{p_k} \left(\frac{1}{x^*_k}\right)^{1-\rho}\ \
\Rightarrow \ \ \frac{a_j}{p_j} > \frac{a_k}{p_k} \\
\forall j \neq k,& 0<x^*_k\le x^*_j<1 & \Rightarrow &\frac{a_j}{a_k} = \frac{p_j}{p_k} \left(\frac{x^*_j}{x^*_k}\right)^{1-\rho}\  \
\Rightarrow \ \ \frac{a_j}{p_j} \ge \frac{a_k}{p_k} 
\end{array}
\end{equation}

\medskip

\noindent{\bf Leontief utility functions:} 
An optimal bundle at Leontief is essentially driven by $a_j$s and not so much by prices. Note that to achieve unit amount of utility
the buyer has to buy at least $a_j$ amount of each good $j$, and therefore has to spend $\sum_j a_jp_j$ money. Thus from money $B$
she can obtain at most $\frac{B}{\sum_j a_jp_j}$ units of utility. Further, 
since she will always buy the cheapest optimal bundle, we get,

\begin{equation}\label{eq.leonob}
\forall j, x^*_j =\beta a_j,\ \  \mbox{where}\ \ \beta=\min\  \left\{\frac{B}{\sum_j a_jp_j},\ \frac{1}{\max_j a_j}\right\} 
\end{equation}

The next theorem follows using the KKT conditions of (\ref{eq.kt}) for each class of utility functions.
\begin{theorem}\label{thm:ob}
Given prices $\pp$ and budget $B$, conditions (\ref{eq.linob}), (\ref{eq.splcob}), (\ref{eq.cesob}) and (\ref{eq.leonob}) together with
feasibility constraints of (\ref{eq.ob}) exactly characterizes $\xx^* = \CB_U(\pp,B)$ for $U\in \CH_{lin}$, $U\in \CH_{splc}$, $U\in
\CH_{ces}$ and $U\in \CH_{leon}$ respectively.
\end{theorem}

\subsection{Linear functions}

Recall that, if $U\in \CH_{lin}$ then $U(\xx)=\sum_j a_j x_j$, where $\sum a_j=1$. 
First we need to figure out which $a_j$s are non-zero. 
\begin{lemma}\label{lem.lin1}
For $p_j=1,\ \forall j$ and $B=n$, if $\xx =\CB_U(\pp,B)$, then $x_j=0\Rightarrow a_j=0$. 
\end{lemma}

\begin{proof}
Since $B=\sum_j p_j$, the agent has enough money to buy all the good completely, and the lemma follows as the agent buys 
cheapest optimal bundle.  \end{proof}

Lemma \ref{lem.lin1} implies that one query is enough to find the set $\{j\ |\ a_j>0\}$. Therefore, wlog we now assume that $\forall j \in
\CG,\ a_j>0$. 

Note that, it suffices to learn the ratios $\frac{a_j}{a_1},\ \forall j\neq 1$ exactly in order to learn $U$, as $\sum_j a_j=1$.  Since
the bit length of the numerator and the denominator of each $a_j$ is at most $n$, we have that $\nfrac{1}{2^{2n}} \le \nfrac{a_j}{a_1}
\le 2^{2n}$.
Using this fact, next we show how to calculate each of these ratios using $O(n)$ revealed preference queries, and in turn the entire
function using $O(dn)$ queries.

Recall the optimality conditions (\ref{eq.linob}) for linear functions.
Algorithm \ref{a:lin} determines $\nfrac{a_j}{a_1}$ when called with $H=2^{2n}$, $q=1$ and $x^e_j=0$.\footnote{These three inputs are irrelevant
for learning linear functions, however they will be used to learn SPLC functions in Appendix \ref{sec:rpqsplc}.}
The basic idea is to always set budget $B$ so low that the agent can by only the most preferred good, and then do binary search by
varying $p_j$ appropriately. Correctness of the algorithm follows using (\ref{eq.linob}) and the fact that bit length of
$\nfrac{a_j}{a_1}$ is at most $2n$. 

\begin{algorithm}[!htb]
   \caption{Learning Linear Functions: Compute $\nfrac{a_j}{a_1}$}
   \label{a:lin}
\begin{algorithmic}
   \STATE {\bfseries Input:} Good $j$, upper bound $H$, quantity $q$ of goods, extra amount $x^e_j$.
   \STATE {\bf Initialize:} $\ L\leftarrow 0$; $\ p_1\leftarrow 1$; $\ p_k \leftarrow 2^{10n}, \ \forall k \in
   \CG\setminus\{j,1\}$; $\ i\leftarrow 0$; $\ \flag\leftarrow nil$
   \WHILE{$i<=4n$}
   	\STATE $i\leftarrow i+1$; $\ p_j \leftarrow\frac{H+L}{2}$; $\ B\leftarrow x^e_j*p_j+\frac{\min\{p_1,p_j\}}{q}$; $\ \xx\leftarrow\CB_U(\pp,B)$
	\STATE {\bf if }{$x_j>0\ \&\ x_1>0$} {\bf then } {\bf Return} $p_j$; 
	\STATE {\bf if} {$x_j>0$} {\bf then} $L\leftarrow p_j$; $\flag\leftarrow 1$; {\bf else} $H\leftarrow p_j$; $\flag \leftarrow 0$;
   \ENDWHILE
   \STATE {\bf if} {$\flag = 1$} {\bf then} Round up $p_j$ to nearest rational with denominator at most $2^n$
   \STATE {\bf else} Round down $p_j$ to nearest rational with denominator at most $2^n$ 
   \STATE {\bf Return} $p_j$.
\end{algorithmic}
\end{algorithm}

\begin{theorem}\label{thm.rplin}
The class $\Hlin$ is learnable from $O(nd)$ revealed preference queries.
\end{theorem}

\subsection{Separable piecewise-linear concave (SPLC) functions}\label{sec:rpqsplc}
In this section we design a learning mechanism for a function of class $\CH_{splc}$, which requires us to learn slops as well as lengths
of each of the segment $(j,k)$. As discussed in Section \ref{ss:classes} that for any $\alpha>0$, $\CB_{U}(\pp,B)=\CB_{\alpha U}(\pp,B),\
\forall(\pp,B)$, it is impossible to distinguish between functions $U$ and $\alpha U$, and that is why we made normalizing assumptions
while defining $\CH_{lin}$, $\CH_{ces}$ and $\CH_{leon}$. Similarly for $U\in \CH_{splc}$ we wlog assume that $a_{11}=1$ now on. 

As the size of each $a_{jk}$ and $l_{jk}$ is at most $n$, we have $\frac{1}{2^n} \le a_{jk} \le 2^n,\ \forall (j,k)$ and
$\frac{1}{2^n} \le l_{jk} \le 2^n,\ \forall j,\ \forall k<|U_j|$; recall that length of the last segment for each good is infinity,
i.e., $l_{j|U_j|}=\infty$.  Therefore, slop of first segments $a_{j1}$ of good $j$, can be learned by calling Algorithm \ref{a:lin}
with $H=2^n$, $q=\frac{1}{2^{n+1}}$ and $B^e=0$; extra budget $B^e$ will be used to learn slops of second segment onward. This will make sure
that no segment can be bought fully during the algorithm, and therefore when a good is bought we know that the allocation is on its
first segment. 

Next we show how to learn length $l_{j1}$ of this segment.
Suppose we fix prices $p_1$ and $p_j$ such that agent is prefers segment $(j,1)$ before $(1,1)$ before $(j,2)$, i.e.,
$\frac{a_{1j}}{p_j}>\frac{a_{11}}{p_1}>\frac{a_{j2}}{p_j}$, then Algorithm \ref{a:splc} outputs $l_{j1}$ when provided with $p_1$, 
$p_j$ and $x^e_j=0$. The basic idea is to do binary search by varying the budget appropriately.

\begin{algorithm}[!htb]
   \caption{Learning Linear Functions: Compute $l_{jk}$}
   \label{a:splc}
\begin{algorithmic}
   \STATE {\bfseries Input:} Good $j$, prices $p_j$ and $p_{1}$, extra amount $x^e_j$.
   \STATE {\bf Initialize:} $\ H\leftarrow 2^{n+1}$; $\ L\leftarrow 0$; $\ p_k \leftarrow 2^{10n}, \ \forall k \in
   \CG\setminus\{j,1\}$; $\ i\leftarrow 0$; $\ \flag\leftarrow nil$;
   \STATE $B\leftarrow (H+x^e_j) *p_j$; $\ \xx\leftarrow \CB_U(\pp,B);\ $ {\bf if} $x_1=0$ {\bf then} {\bf Return} $\infty$;
   \WHILE{$i<=2n+1$}
   	\STATE $i\leftarrow i+1$; $\ T\leftarrow \frac{H+L}{2}$; $\ B\leftarrow (T+x^e_j)*p_j$; $\ \xx\leftarrow \CB_U(\pp,B)$;
	\STATE {\bf if $x_1=0$} {\bf then} $L\leftarrow T$; $\flag\leftarrow 1$; {\bf else} $H\leftarrow T$; $\flag \leftarrow 0$;
   \ENDWHILE
   \STATE {\bf if} $\flag=1$ {\bf then} Round up $T$ to nearest rational with denominator at most $2^n$;
   \STATE {\bf else} Round down $T$ to nearest rational with denominator at most $2^n$;
   \STATE {\bf Return} $T$;
\end{algorithmic}
\end{algorithm}

The next question is what should be $p_1$ and $p_j$ so that $\frac{a_{1j}}{p_j}>\frac{a_{11}}{p_1}>\frac{a_{j2}}{p_j}$ is
ensured. Setting $p_j=a_{j1}$ and $p_1=a_{11}+\epsilon>1$ ensures $\frac{a_{11}}{p_1}<\frac{a_{1j}}{p_j}$. Further, $\epsilon=\frac{1}{2^{2n+1}}$
ensures $\frac{a_{11}}{p_1}>\frac{a_{j2}}{p_j}$ using the next claim.

\begin{claim}\label{cl.splc}
If $\epsilon=\frac{1}{2^{2n+1}}$, $p_1=a_{11}+\epsilon$, and $p_j=a_{jk}$ then $\frac{a_{j(k+1)}}{p_j} < \frac{a_{11}}{p_1}$.
\end{claim}
\begin{proof}
As $a_{jk} > a_{j(k+1)} \ge \frac{1}{2^n}$ with bit length of both being at most $n$, we have $a_{jk}-a_{j(k+1)} \ge \frac{1}{2^n}$. 
\[
\begin{array}{r}
a_{jk}-a_{j(k+1)} \ge \frac{1}{2^n}\ \Leftrightarrow\ 1-\frac{a_{j(k+1)}}{a_{jk}} \ge \frac{1}{2^{2n}}\ \Leftrightarrow\
\frac{a_{j(k+1)}}{a_{jk}} \le 1-\frac{1}{2^{2n}}\\ \Leftrightarrow\  \frac{a_{j(k+1)}}{p_j} < \frac{1}{1+\epsilon}\ \Leftrightarrow\
\frac{a_{j(k+1)}}{p_j} < \frac{a_{11}}{p_1}
\end{array}
\]
\end{proof}

\noindent{\bf Induction.} Once we learn slops and lengths of up to $k^{th}$ segment of good $j$, can learn $a_{j(k+1)}$ by calling Algorithm
\ref{a:lin} with $H=a_{jk}$, $q=\frac{1}{2^{n+1}}$ and $x^e_j=\sum_{t\le k}l_{jt}$. And then learn $l_{j(k+1)}$ by calling Algorithm
\ref{a:splc} with $p_j=a_{j(k+1)}$, $p_1=1+\epsilon$, and $x^e_j=l_{jk}$ (it works using Claim \ref{cl.splc}). We stop when Algorithm
\ref{a:splc} returns $\infty$.
\medskip

For each good $j\neq 1$, think of a hypothetical $0^{th}$ segment with $a_{j0}=2^{n+1}$ and $l_{j0}=0$, and apply the above inductive
procedure to learn $a_{jk}$ and $l_{jk}$ for all $1\le k\le |U_j|$. To learn the parameters for good $1$, we can swap its identity with
some other good, and repeat the above procedure. The number of calls to oracle $\CB_U$ in Algorithms \ref{a:lin} and \ref{a:splc} are
of  $O(n)$, and if there are at most $\kappa$ segments in each $U_j$, then total sample complexity for learning such an SPLC function
is $O(nd\kappa)$.

\begin{theorem}
The class $\CH_{splc}$ is learnable from $O(nd\kappa)$ revealed preference queries.
\end{theorem}

\subsection{CES and Leontief functions}
In this section we show that surprisingly constantly many queries are enough to learn a CES or a Leontief function. The reason behind
this is that the optimal bundles of these functions are well behaved, e.g., a buyer buys all the goods of non-zero amount, and in a
fixed proportion in case of Leontief functions. 
\medskip

\noindent{\bf CES. }
Let $U\in \CH_{ces}$ be a function defined as $U_{\aa\rho}(\xx)=(\sum_j a_j x_j^\rho)^{\nfrac{1}{\rho}}$, where $\sum_j a_j=1$ and $\rho<1$.
Since an optimal bundle for such a $U_{\aa\rho}$ contains non-zero amount of good $j$ only if $a_j>0$, wlog we assume that $a_j>0,\ \forall j$.
We show that two queries, with prices $\pp>0$ and budget $B<\min_j p_j$ are enough to learn $\Hces$ from revealed preference queries.

Let $p^1_j=1,\ \forall j$, $p^2_j=j,\ \forall j$, $B=0.5$, $\xx^1=\CB_U(\pp^1,B)$ and $\xx^2=\CB_U(\pp^2,B)$. 
Since there is not enough budget to buy any good completely in either query, we have $0<x^i_j<
1,\ i=1,2, \forall j$, using (\ref{eq.cesob}). Like for linear functions, it is enough to learn ratios $\frac{a_j}{a_1},\ \forall j \neq
1$. Using Equation (\ref{eq.cesob}), Section \ref{ss:characterization_opt_bundles},  we get the following.
\[
\begin{array}{lcl}
i=1,2, \forall j\neq 1,\ \frac{a_j}{a_1}=\frac{p^i_j}{p^i_1} \left(\frac{x^i_j}{x^i_1}\right)^{1-\rho}& \Rightarrow &
\left(\frac{x^1_j}{x^1_1}\right)^{1-\rho} = j \left(\frac{x^2_j}{x^2_1}\right)^{1-\rho} \\
& \Rightarrow & (1-\rho) \log{\frac{x^1_j}{x^1_1}} =
\log{j} + \log{\frac{x^2_j}{x^2_1}} 
\end{array}
\]

Since $\xx^1$ and $\xx^2$ are known, we can evaluate the above to get $\rho$ and $\nfrac{a_j}{a_1}$. 
\medskip

\noindent{\bf Leontief. }
Consider a Leontief function $U_{\aa}\in \CH_{leon}$ such that $U_{\aa}(\xx)=\min_j \nfrac{x_j}{a_j}$, where $\sum_j a_j=1$. Wlog, we assume that
$a_j>0,\ \forall j$; if $a_j=0$ then $x_j=0$ in an optimal bundle at any given prices and budget. We show that one query, with prices
$\pp>0$ and budget $B< \min_j p_j$, is enough to determine $U_{\aa}$ and thus one query suffices to learn the class $\Hleon$ from revealed preference queries. 

Suppose, $\xx=\CB_U(\pp,B)$ where $p_j=1,\ \forall j$ and $B=0.5$. Then using (\ref{eq.leonob}), we get $\beta=\nfrac{B}{\sum_j a_j} =
B= 0.5$ and $a_j=\nfrac{x_j}{\beta}=2x_j$.  

\begin{theorem}\label{thm.rp3}
The classes $\Hces$ and $\Hleon$ are learnable from $O(1)$ revealed preference queries.
\end{theorem}

 \subsection*{Acknowledgments}
 This work was supported in part by AFOSR grant FA9550-09-1-0538, ONR grant
 N00014-09-1-0751, NSF grants CCF-0953192 and CCF-1101283, a Microsoft
 Faculty Fellowship, and a Google Research Award.

\bibliographystyle{abbrv}
\bibliography{revealedpreferencebib}

\newpage
\appendix
\section*{Appendix}

\section{Multi-class learning background}
\label{ss:background_multi}

 Here we review previously established results on multiclass learnability, that are relevant to the results in our paper.
 
 \subsection{The new bound for linear classes}
 In our work, we employ the following recent upper bound by \cite{DanielyS14}  on the sample complexity of \emph{$D$-dimension linear} hypothesis classes (Definition \ref{d:D_dim_linear_class}). 
 
 \begin{theorem}[\cite{DanielyS14}, Theorem 5, part 1]
  For every $\Psi:\Xcal \times \Ycal \to \reals^D$, the (PAC) sample complexity of learning $\CH_\Psi$ is 
  \[
   m[\CH_\Psi](\epsilon,\delta) = O\left(\frac{D\log(1/\epsilon)+\log(1/\delta)}{\epsilon} \right).
  \]
 \end{theorem}
 
 The upper bound in the above Theorem is achieved by a compression scheme based algorithm.  
 That is, the authors show that there always exists a compression scheme for linear classes, which yields learnability for both the realizable and the agnostic case as we outline next.

 \subsection{Compression scheme based learning}
 \label{ss:compression_schemes}
 
 \begin{definition}[Compression scheme]
  Let $\CH\subseteq \Ycal^\Xcal$ be a hypothesis class.
  A compression scheme of size $d$ for the class $\CH$ consists of two functions $C: \bigcup_{n\in \naturals}(\Xcal \times \Ycal)^n \to (\Xcal \times \Ycal)^d $ and $D :(\Xcal \times \Ycal)^d \to \CH$ satisfying the following condition:
  \begin{itemize}
   \item Let $S = ((x_1, y_1) , \ldots , (x_n, y_n))$ with $y_i = h(x_i)$ for some $h\in \CH$ and all $i$. Then $C(S)$ is a subsequence of $S$ and for the function $h_D = D(C(S))\in \CH$ we have 
   $
    h_D(x_i) = y_i
   $
   for all $x_i$ in $S$.
  \end{itemize}
 \end{definition}

 If a class admits a compression scheme, then it is learnable both in the realizable and in the agnostic case with the following sample complexity bounds (also see \cite{MLbook}, Chapter 30):
 
 \begin{theorem}[Based on \cite{LW86}]
  Assume that class $\CH \subseteq \Ycal^\Xcal$ has a compression scheme $(C,D)$ of size $d$. 
  Then it is learnable in the realizable case (by the algorithm $D \circ C$)
  with sample complexity  satisfying
  \[
   m[\CH](\epsilon,\delta) = O\left(\frac{d \log(1/\epsilon) +  1/\delta}{\epsilon}\right).
  \]
  Moreover, the class is also learnable in the the agnostic case with sample complexity satisfying
  \[
   m[\CH](\epsilon,\delta) = O\left(\frac{d \log(d/\epsilon) +  1/\delta}{\epsilon^2}\right).
  \]
 \end{theorem}

 \subsection{Lower bounds}
  The following measure of complexity of a hypothesis class yields a lower bound for multi-class learnability:

 \begin{definition} [N-shattering; Natarajan dimension]
  A set $\{x_1, \ldots, x_n\}$ is $N$-shattered by a class of functions $\CH \subseteq \Ycal^\Xcal$ if there exists two functions $f_1, f_2\in \Ycal^\Xcal$ with $f_1(x_i)\neq f_2(x_i)$ for all $i\in [n]$, such that, for any binary vector $v \in \{0,1\}^n$ of indices, there exists an $h_v \in H$ with
 \[
  h_v(x_i) \left\lbrace
  \begin{array}{lll}
  = f_1(x_i) & ~\text{if}~ & v_i = 1\\ 
  = f_2(x_i) & ~\mbox{if}~ & v_i = 0\\
  \end{array}
  \right.
  \]
  We call the size of a largest $N$-shattered set the \emph{Natarajan-dimension} of the class $\CH$.
  \end{definition}

  \begin{theorem}[\cite{Natarajan89}]\label{nat_graph_upper_lower}
     The sample complexity of learning a multi-class hypothesis class $\CH$ satisfies
  \[
 m[\CH](\epsilon,\delta) =  \Omega\left(\frac{d_N(H) + \ln(1/\delta)}{\epsilon}\right) 
  \]
  \end{theorem}

\section{The lower bound in Theorem \ref{t:linear_result}}
\label{s:lower_bound_for_Hlin}

We show a lower bound on the Natarajan dimension of $\widehat{\Hlin}$:

\begin{lemma}
\label{l:Nidim}
 The Natarajan dimension of the class $\widehat{\Hlin}$ is at least $d-1$.
\end{lemma}
\begin{proof}
 We show that there is a set of pairs of price vectors and budgets of size $d-1$ that is $N$-shattered by $\widehat{\Hlin}$.
 Consider the set $\{(\pp^1, 1) \ldots (\pp^{d-1}, 1)\}$ with all budgets set to $1$ and with the price vectors defined by:
 \begin{equation*}
 p^j_i = \left\{
\begin{array}{ll}
 1                           & \text{if} \quad i=1 \\
 1                           & \text{if} \quad i=j \\
 10                         & \text{otherwise}
\end{array}
\right. 
\end{equation*}
We consider the following functions $f_0$ and and $f_1$ that map the pairs $(\pp^j, 1)$ to bundles.
We set $f_0(\pp^j, 1) = (1, 0, \ldots, 0)$ for all $j$; that is, $f_0$ maps all pairs to the bundle where only the first good is bought.
Now we define $f_1$ by setting the $i$-th coordinate of the bundle $f_1(\pp^j, 1)$ to
 \begin{equation*}
 (f_1(\pp^j, 1))_i = \left\{
\begin{array}{ll}
 1                           & \text{if} \quad i=j \\
 0                         & \text{otherwise}
\end{array}
\right. 
\end{equation*}
That is, $f_1$ maps $(\pp^j, 1)$ to the bundle where only the $j$-th good is bought.

Now, given a vector $\vv\in\{0,1\}^{d-1}$, the demand function that is defined by the utility vector $\ww$ with
 \begin{equation*}
 w_i = \left\{
\begin{array}{ll}
 1                           & \text{if} \quad i=1 \\
 2                           & \text{if} \quad v_{i-1} =1 \\
 1                           & \text{if} \quad v_{i-1} = 0
\end{array}
\right. 
\end{equation*}
yields $\widehat{U_{\ww}}(\pp^i, 1) =  f_{v_j}(\pp^i, 1)$ for all $i,j$. Thus, the set $\{(\pp^1, 1) \ldots (\pp^{d-1}, 1)\}$ is $N$-shattered.
\end{proof}

According to Theorem \ref{t:linear_result} above, this lower bound on the Natarajan dimension yields the lower bound for learning $\widehat{\Hlin}$ stated in the Theorem.
It is not difficult to see that the shattering construction in the above lemma also yields the same lower bound for learning $\Hlin$ in the revealed preference model. For this, observe that the two functions $f_0$ and $f_1$ in the construction not only yield different optimal bundles on the $(\pp^j, 1)$,  but these optimal bundles also have different utility values.

\section{Statistical learning from revealed preferences}
\label{s:SPCL_CES_Leon}

\subsection{SPLC functions}
Recall that an SPLC utility function $U_{\AA\LL}$ can be defined by two $d\times \kappa$ matrices. Entry $a_{ij}$ of $\AA$ stands for the slope of the $j$-th segment of $U_i$ (the piecewise linear function for the marginal utility over good $i$).
Entry $l_{ij}$ of $\LL$ is the length of that same segment.
If the maximum number of segments and their lengths are known a priori, we can employ the same technique as for learning linear utility functions from revealed preferences.
That is, let $\Hsplc^{\LL}$ denote the subclass of all SPLC functions where number and lengths of the segments are fixed (defined by matrix $\LL$).

As for linear utility functions, we can identify admissible candidates for optimal bundles.
Note that, an agent will greedily buy segments according to an order of $a_{ij}/p_i$ and in an optimal bundle all, but at most one, segments are bought fully (see also Section \ref{ss:characterization_opt_bundles}).
Thus, here we call a bundle $\xx$ admissible for some $(\pp,B)$ if $|\{j ~:~ x_j \neq \sum_{g\leq h} l_{jg} \text{ for some } h \in [d]\}| \leq 1$ and $\inner{\pp, \xx} = B$.
As in the linear case, if $\inner{\pp, \bo_d} = \sum_{i\in\CG} p_i \leq B$, we also call the all $1$-bundle $\bo_d$ admissible (and in this case, it is the only admissible bundle).

Then the corresponding class of demand functions $\widehat{\Hsplc^{\LL}}$ is a $\kappa d$-dimensional linear class as witnessed by the mapping
\begin{equation*}
\Psi((\pp, B), \xx) = \left\{
\begin{array}{l}
 \xx^{\kappa d} \text{ if admissible} \\
 \bz_{\kappa d} \text{ otherwise },
\end{array}
\right. 
\end{equation*}
where $\xx^{\kappa d}$ is the ``split'' of $\xx$ into $\kappa d$ dimensions according to the matrix $\LL$ as follows:
\begin{equation*}
x^{\kappa d}_i = \left\{
\begin{array}{ll}
 l_{hj}                           & \text{if}\quad \sum_{g \leq j} l_{hg} \leq x_h \quad\text{and}\quad h = \lceil{i/d}\rceil  \quad\text{and}\quad i = j \mod d \\
 x_j -  \sum_{g \leq j-1} l_{hg}  & \text{if}\quad \sum_{g \leq j-1} l_{hg} \leq x_h <  \sum_{g \leq j} l_{hg} \quad\text{and}\quad h = \lceil{i/d}\rceil  \quad\text{and}\quad i = j \mod d \\
 0                                & \text{if}\quad x_h < \sum_{g \leq j-1} l_{hg}    \quad\text{and}\quad h = \lceil{i/d}\rceil  \quad\text{and}\quad i = j \mod d \\
\end{array}
\right. 
\end{equation*}

Therefore, we immediately get the sample complexity result:

\begin{theorem}
 The classes $\Hsplc^{\LL}$ of linear utility functions with known segments are learnable efficiently in the revealed preference model with sample complexity   
  \[
   O\left(\frac{\kappa d\log(1/\epsilon)+\log(1/\delta)}{\epsilon} \right).
  \] 
\end{theorem}

In order to argue for the computational efficiency in the above theorem, according to Remark \ref{r:svm_efficient}, we need to show how to compute the second best {\em
admissible} bundle in polynomial-time. As in the linear case, if $\bo_d$ is an admissible bundle (that is, if $\inner{\pp,\bo_d} \leq B$), then any other bundle is second best (with respect to the mapping $\Psi$).

Otherwise, for given $(\pp,B)$, we design an $O( d)$-time algorithm to compute the
second best admissible bundle, i.e., $\yy \in \argmax_{\xx \mbox{ admissible }, \xx\neq \xx^*} \aa\cdot \xx$, where $\xx^*$ is the
optimal bundle.

Similar to an optimal bundle for a function of $\Hlin$, an optimal bundle for $U_{\aa}\in \Hsplc^{\LL}$ can be computed by sorting segments
$(j,k)$ in decreasing order of $\frac{a_{jk}}{p_j}$ and buying them in order (Section \ref{ss:characterization_opt_bundles});
$a_{jk}>a_{j(k+1)}$ ensures that segments of a good are bought from first to last. Thus, the second best admissible optimal bundle can also
be computed in similar way as done in Section \ref{s:computing_second_best} for $\Hlin$. 

Corresponding to the optimal bundle $\xx^*$, let $x^*_{jk}$ denote the allocation on segment $(j,k)$. Let $k_j$ be the last segment
bought of good $j$. Then, clearly there exists exactly one good, say $t$, such that $x^*_{tk_t}<l_{tk_t}$.  Let $\yy$ be the second
best admissible bundle. Like in Lemma \ref{le:1} it follows that $y_l >x*_l$ for exactly one good $l$. Further, the extra allocation
has to be on segment $(l,(k_l+1)$ if $l\neq t$ else $(t,k_t)$. Next like Lemma \ref{le:2} $y_i<x^*_i$ for exactly one good $i$, and
decrease in allocation is on segment $(i,k_i)$. Finally, similar to Lemma \ref{le:3}, if $l\neq t$ and $x^*_{tk_t}<l_{tk_t}$ then
$y_t<x^*_t$, and if $x^*_{tk_t}=l_{tk_t}$ then $l$ has to be the good whose was the first to be not allocated.

Thus, the algorithm to compute second best bundle will have check only $O(d)$ bundles, namely if $x^*_{tk_t}<l_{tk_t}$ then either
money is transferred from one of segments $(i,k_i)$ to one of segments $(t,k_t)$ or from $(t,k_t)$ to one of $(l,(k_l+1)),\ \forall
l\neq t$, and otherwise from $(i,k_i)$ to the best unallocated segment in $\xx^*$.

\subsection{CES known $\rho$}

We show that the classes of demand functions $\widehat{\Hces^\rho}$ are also $d$-dimensional linear classes, for any $\rho \in\Rplus, \rho\leq 1$ (the case $\rho =1$ yields linear utility functions whose demand functions were shown to be $d$-dimensional linear above).

Recall that a CES function is defined by a parameter $\rho \in\Rplus, \rho\leq 1, \rho\neq 0$ and a vector $\aa\in \Rplus^d$.
Note that for some price vector $\pp$ and budget $B$, we have
$$
\argmax_{\xx\in[0,1]^n, \inner{\pp, \xx} \le B} (\sum_j a_j x_j^\rho)^{\nfrac{1}{\rho}} = \argmax_{\xx\in[0,1]^n, \inner{\pp, \xx} \le B} \sum_j a_j x_j^\rho
$$
Thus, we can employ the following mapping:
\begin{equation}
\Psi((\pp, B), \xx) = \left\{
\begin{array}{l}
 \xx^{\rho} \text{ if } \pp\cdot \xx \le B \\
 \bz_{d} \text{ if } \pp\cdot \xx > B,
\end{array}
\right. 
\end{equation}
where $\xx^{\rho} = (x_1^\rho, \ldots, x_d^\rho)$. This yields:

\begin{theorem}
 The classes $\Hces^\rho$ of linear utility functions with known parameter $\rho$ are learnable in the revealed preference model with sample complexity   
  \[
   O\left(\frac{d\log(1/\epsilon)+\log(1/\delta)}{\epsilon} \right).
  \] 
\end{theorem}

\subsection{Leontief}

Learning the class of Leontief functions from revealed preferences in a statistical setting is trivial, since observing \emph{one} optimal bundle reveals all the relevant information (see Section \ref{ss:characterization_opt_bundles}).
 
\section{Statistical learning of utility functions}
\label{s:statistical_value}

As a point of comparison, we also analyze the learnability of classes of utility functions in the standard statistical multi-class learning model (Definition \ref{d:learnability}).
That is, here the input to the learner is a sample $S = ((\xx_1, U(\xx_1)), \ldots, (\xx_n, U(\xx_n)))$ of pairs of bundles and values generated by a distribution $P$ over bundles and labeled by a utility function $U$ from a class $\CH$.
The learner outputs a function from bundles to values $\Acal(S): [0,1]^d \to \Rplus$. 
\subsection{Linear}

Linear functions are learnable in the multi-class learning framework.
The following result has been implicit in earlier works. For completeness, we provide a proof here.

\begin{theorem}
The class of linear functions $H=\{\xx \mapsto \inner{\xx, \ww} ~:~ \ww\in \reals^d\}$ on $\reals^d$ is learnable
with sample complexity $O\left(\frac{d \log(d/\epsilon) + \log(1/\delta)}{\epsilon}\right)$.
\end{theorem}

\begin{proof}[Sketch]
 Note that for any two linear functions $\ww$ and $\ww'$, the set of points on which $\ww$ and $\ww'$ have the same value forms a linear subspace.
 Thus, the the set $H\Delta H$ of subsets of $\Xcal$ where two linear functions $\ww$ and $\ww'$ disagree is exactly the collection of all complements of linear subspaces.
 The set of all linear subspaces of a vector space of dimension $d$ has VC-dimension $d$. 
 Since a collection of subsets has the same VC-dimension as the collection of corresponding complements of subsets, $H\Delta H$ has VC-dimension $d$ for the class $H$ of linear functions.
 
 An \iid sample of size $O\left(\frac{d \log(d/\epsilon) + \log(1/\delta)}{\epsilon}\right)$ is an $\epsilon$-net for $H\Delta H$ with probability at least $1-\delta$ \cite{HausslerW86}.
 This guarantees that (with probability at least $1-\delta$) every function that is consistent with the sample has error at most $\epsilon$. 
 Note that, to find a function $\ww$ that is consistent with a sample, it suffices to find a maximal linearly independent set of vectors $\xx_i$ in the sample.
 The value on a new example can then be inferred by solving a linear system.
\end{proof}

\subsection{SPLC and CES}

It is straightforward to see that learning the class $\Hsplc^{\LL}$ of SPLC utility functions where the number and lengths of the segments are known reduces to learning $\kappa d$-dimensional linear functions, where $\kappa$ is the maximum number of segments per good. For this, given a sample $S$, create a new sample $S^{\kappa d}$ by mapping every example $(\xx, U(\xx))\in S$ to an example $(\xx^{\kappa d}, U(\xx))$ for $S^{\kappa d}$, where $\xx^{\kappa d}\in [0,1]^{\kappa  d}$ is defined coordinate-wise as follows:
\begin{equation*}
x^{\kappa d}_i = \left\{
\begin{array}{ll}
 l_{hj}                           & \text{if}\quad \sum_{g \leq j} l_{hg} \leq x_h \quad\text{and}\quad h = \lceil{i/d}\rceil  \quad\text{and}\quad i = j \mod d \\
 x_j -  \sum_{g \leq j-1} l_{hg}  & \text{if}\quad \sum_{g \leq j-1} l_{hg} \leq x_h <  \sum_{g \leq j} l_{hg} \quad\text{and}\quad h = \lceil{i/d}\rceil  \quad\text{and}\quad i = j \mod d \\
 0                                & \text{if}\quad x_h < \sum_{g \leq j-1} l_{hg}    \quad\text{and}\quad h = \lceil{i/d}\rceil  \quad\text{and}\quad i = j \mod d \\
\end{array}
\right. 
\end{equation*}
Now, we can just learn a linear function $\ww \in \reals^d$ on $S^{\kappa d}$ and predict according to this function (employing the same mapping on a test example).

Similarly, we can reduce learning $\Hces^\rho$ of learning CES functions with fixed parameter $\rho$ to learning linear utility functions.
For this, given a sample $S$, create a new sample $S^{\rho}$ by mapping every example $(\xx, U(\xx))\in S$ to an example $(\zz, (U(\xx))^{\rho})$ for $S^{\rho}$, where $\zz\in [0,1]^{d}$ is defined coordinate-wise by setting ${z}_i  = (x_i)^\rho$.

\subsection{Leontief}

We now show that the class of Leontief functions is learnable. 
Recall that, a Leontief utility function is defined by a vector $\aa = (a_1,\ldots, a_d)$ by $U_{\aa}(\xx)=\min_{j \in \CG} \nfrac{x_j}{a_j}$. 

Note that, given an example $(\xx, y) = (\xx, U_{\aa}(\xx))$, we have
\[
 U_{\aa}(\xx) ~\leq~ \frac{x_j}{a_j}
\]
for all $j\in[d]$ with equality for at least one index $j$. Equivalently, we have
\[
 a_j ~\leq~ \frac{x_j}{U_{\aa}(\xx)}
\]
for all $j\in[d]$ with equality for at least one index $j$. That is, each example provides us with upper bounds on all the (unknown) parameters $a_j$ of the utility function. This suggests the following learning procedure: Going over all training examples, we maintain estimates $b_i$ of the $a_i$, by using the above inequalities (see Algorithm \ref{a:learning_leon}).

\begin{algorithm}[tb]
   \caption{Learning Leontief}\label{a:learning_leon}
   \label{a:leontief_stat_value}
\begin{algorithmic}
   \STATE {\bfseries Input:} Sample $S=((\xx^1, y^1),\ldots, (\xx^m, y^m))$
   \STATE $b_j \gets \min\{b_j, x^1_j/y^1\}$
     \FORALL{$i\in [m]$}
        \FORALL{$j \in [d]$}
          \STATE $b_j \gets \min\{b_j, x^i_j/y^i\}$
        \ENDFOR
     \ENDFOR
   \STATE {\bfseries Return:} vector $\bb = (b_1, \ldots, b_d)$ 
\end{algorithmic}
\end{algorithm}

On a new example, we predict with the Leontief utility function defined by $\bb$.

In order to prove that the above algorithm is a successful learner, we use the following claim, that characterizes the cases where an estimate $\bb$ of a target Leontief function $\aa$ errs on an example $\xx$.

\begin{claim}\label{cl:leontief_error}
 Let $\aa$ and $\bb$ be two vectors (defining Leontief utility functions) with $b_i \geq a_i$ for all $i\in[d]$.
 Then, $U_{\bb}(\xx) \neq U_{\aa}(\xx)$ implies 
 \[
  \frac{x_k}{U_{\aa}(\xx)} < b_k
 \]
 for the index $k$ that defines $U_{\bb}(\xx)$ (that is, the $k$ that minimizes $\nfrac{x_k}{b_k}$).
\end{claim}

\begin{proof}
 Let $\xx$ be some bundle with $U_{\bb}(\xx) \neq U_{\aa}(\xx)$, that is $\min_{j \in \CG} \nfrac{x_j}{b_j} \neq \min_{j \in \CG} \nfrac{x_j}{a_j}$.
 Let $k$ be the index that minimizes the left hand side (that defines $U_{\bb}(\xx)$) and let $i$ be the index that minimizes the left hand side (that defines $U_{\aa}(\xx)$). Then the above inequality implies that either $i\neq k$ or $i = k$ and $a_i = a_k \neq b_k$.
 
 If  $i = k$ and $a_i \neq b_k$, then we get
 \[
  \frac{x_k}{U_{\aa}(\xx)} ~=~ \frac{x_k}{\nfrac{x_k}{a_k}} ~=~ a_k ~<~ b_k,
 \]
 by the assumption that $b_i \geq a_i$ for all $i\in[d]$. 
 If $i\neq k$, we have
 \[
  \frac{x_k}{b_k} ~<~ \frac{x_i}{b_i} ~\leq~ \frac{x_i}{a_i} ~=~ U_{\aa}(\xx),
 \]
 and thus
 \[
 \frac{x_k}{U_{\aa}(\xx)} ~<~ b_k.
 \]
\end{proof}

\begin{theorem}
The class of Leontief utility functions is learnable
with sample complexity $O\left(\frac{d \log(d/\delta)}{\epsilon}\right)$.
\end{theorem}

\begin{proof}
 We show that Algorithm \ref{a:leontief_stat_value} is a successful learner for the class of Leontief utility functions.
 Let $\aa$ be the vector that defines the target Leontief function. For each $j \in [d]$, we define consider an interval $[a_j, B_j]$, where $B_j$ is defined by
 \[
  B_j := \min\{B\in \reals ~:~ \Pr_{x \sim P}[ (x_j/U_{\aa}(\xx)) \in [a_j, B] ] \geq \epsilon/d \}.
 \]
 Note that we may have $B_j = a_j$, in which case the interval contains only one point.
 Claim \ref{cl:leontief_error} implies that any Leontief utility function defined by a vector $\bb$ with $b_j \in [a_j, B_j]$ for all $j$ has error at most $\epsilon$ since for any $b_j \leq B_j$ we have
 \[
  \Pr_{x \sim P}[ (x_j/U_{\aa}(\xx)) < b_j ] ~\leq~ \epsilon/d
 \]
 by definition of $B_j$.
 Thus, it suffices to show that the vector $\bb$ that is returned by Algorithm \ref{a:leontief_stat_value} satisfies this requirement (with high probability).
 
 Consider a sample $S=((\xx^1, y^1),\ldots, (\xx^m, y^m))$, with instances generated \iid by the distribution $P$ over bundles and labeled by Leontief function $\aa$ (that is $y^i = U_{\aa}(\xx^i)$). The output vector $\bb$ satisfies $b_j \in [a_j, B_j]$ for all $j$ if, for every index $j$, there exists an example $\xx^i$ in the sample with $x^i_j/y^i = x^i_j/U_{\aa}(\xx) \leq B_j$, that is if the sample $S$ hits all the intervals $[a_j, B_j]$. By definition of $B_j$, the probability that an \iid sample from $P$ of size $m$ does not hit all the intervals is bounded by
 \[
  n (1 - \nfrac{\epsilon}{n})^m ~\leq ~ \e^{\frac{\epsilon m}{d}}.
 \]
 If $m \geq \frac{d \ln(d/\delta)}{\epsilon}$, this probability is bounded by $\delta$.
 Thus, we have shown that with probability at least $1-\delta$ over the training sample $S$ algorithm \ref{a:leontief_stat_value} outputs a Leontief function of error at most $\epsilon$.
\end{proof}

\section{Learning Utility Functions via Value Queries}
\label{ss:value_query}
In this section we show how to learn each of utility functions $\CH_{lin}$, $\CH_{splc}$, $\CH_{ces}$ and $\CH_{leon}$ efficiently from value queries.
In the value query learning setting, a learning algorithm has access to an oracle that, upon given the input of a bundle $\xx$ , outputs the corresponding value $U(\xx)$ of some utility function $U$.
Slightly abusing notation, we also denote this oracle by $U$.

\begin{definition}[Learning from value queries]
 A learning algorithm \emph{learns a class $\CH$ from $m$ value queries}, if for any function $U\in\CH$, if the learning algorithm is given responses from oracle $U$, then after at most $m$ queries the algorithm outputs the function $U$. 
\end{definition}

The complexity of a query learning  algorithm is measured in terms of the number of queries it needs to learn a class $\CH$.
It is considered efficient if this number is polynomial in the size of the target function.
Since we assume that all defining parameters in the classes of Section \ref{ss:classes} are numbers of bit-length at most $n$, we will show that $\poly(n,d)$ queries suffice to learn these classes.

\medskip

\noindent{\bf Linear function. }
For a function $U_{\aa}\in \CH_{lin}$, where $U_{\aa}(\xx)=\sum_j a_j x_j$, $d$ queries are enough to determine it. Define $\forall k\le d,\ 
x^k_j=0,\ \forall j\neq k$ and $x^k_k=1$. Then clearly, $a_k=U(\xx^k)$. 
\medskip
\medskip

\noindent{\bf SPLC function. }
Given a function $U\in \CH_{splc}$ it can be decomposed as $U(\xx)=\sum_j U_j(x_j)$, where each $U_j$ is a piecewise-linear concave
function. As described in Section \ref{ss:classes}, each $U_j$ constitutes of a set of pieces with slopes and lengths. We will learn each
such $U_j$ separately. Let $a_{jk}$ be the slope of segment $k$, and $l_{jk}$ be its length. Let $r$ be the number of segment in
function $U_j$, then except for $l_{jr}$ (which is $\infty$) let $n$ be the maximum bit length of 
any $a_{jk}$ or $l_{jk}$, then $\nfrac{1}{2^n} \le a_{jk},l_{jk} \le 2^n$.
Note that $r$ is unknown.

Given lengths and slopes of segments $1,\dots k-1$ determining the slope of segment $k$ is easy: let $L=\sum_{s<k} l_{js}$ and ask for
$x_j=L+\epsilon$, where $\epsilon<\nfrac{1}{2^n}$. Then $U_j(x_j)=\sum_{s<k} a_{js}l_{js} + a_{jk} \epsilon$ (as
$\epsilon<l_{jk}$) gives the value of
$a_{jk}$. Let $u_L = \sum_{s<k} a_{js} l_{js}$. 

Next is to learn the length $l_{jk}$ of $k^{th}$ segment. Note that, $k$ is the last segment of function $U_j$ if and only if
$U_j(L+2^{n+1})= u_L + a_{jk} 2^{n+1}$. This is because if it is not the last segment then $l_{jk}\le 2^{n}$. Thus,
one query is enough to check this. Suppose $k$ is not the last segment, then we will compute $l_{jk}$ through a binary search, as
follows:

\begin{itemize}
\item[$S_1$] Let $l_l = 0$ and $l_h=2^{n+1}$. Set $i=0$.  
\item[$S_2$] Set $l=\frac{l_l+l_h}{2}$ and $x_j = L + l$. 
\item[$S_3$] If $U_j(x_j) < u_l + u_{jk} l$ then set $l_h=l$, else set $l_l=l$. 
\item[$S_4$] Set $i=i+1$. If $i>2n$ then output $l$ and exit. Else go to $S_2$.
\end{itemize}

In the above procedure we maintain the invariant that $l_l \le l\le l_h$. 
In step $S_3$ of an iteration, the inequality holds only if $l_{jk}<l$, and therefore the $l_h$ is reset to $l$. 
The correctness of the procedure follows from the fact that bit length of $l_{jk}$ is at most $n$.

We learn each $U_j$ separately starting from first to the last segment. This requires $n$ queries to learn each $l_{jk}$, and one query
to learn $a_{jk}$, thus total of $O(n|U_j|)$ queries. Function $U\in \CH_{splc}$ can be learned by making $O(n\kappa d)$ queries to its
value oracle, where $\kappa =\max_j |U_j|$.
\medskip
\medskip

\noindent{\bf CES function with known $\rho$. }
Let $U\in \CH^{\rho}_{ces}$ such that $U(\xx)=(\sum_j a_j x_j^\rho)^{1/\rho}$, where $\rho$ is given. Learning such a function is
equivalent to learning a linear function. Thus for $\xx^k$ as defined in case of Linear functions, we get $a_k=U(\xx^k)^{1/\rho}$. 
\medskip

\noindent{\bf Leontief function. }
Let $U \in \CH_{leon}$ such that $U(\xx)=\min_j \nfrac{x_j}{a_j}$, where every bit length of every $a_j$ is at most $n$. 
In other words, if $a_j>0$ then $\frac{1}{2^n} \le a_j \le 2^n$. Therefore, given that $a_j,a_k>0$, we have $\frac{1}{2^{2n}}\le
\frac{a_k}{a_j} \ge 2^{2n} ,\ \forall j,k$. 

Since $\frac{0}{0}$ is considered as $\infty$, for $x_j=0$ and $\forall k\neq j,\ x_k=1$, $U(\xx)>0$ if and only if $a_j=0$. 
Thus we can figure out all the non-zero $a_j$s using $d$ queries, and therefore wlog assume that $a_j>0,\ \forall j$. Consider a bundle
$\xx^k$, where $x^k_j = 1,\ \forall j\neq k$, and $x^k_k < \nfrac{1}{2^{2n}}$.
\[
\forall j \neq k,\ 
\frac{a_k}{a_j} \ge \frac{1}{2^{2n}} \Rightarrow \frac{1}{a_j} \ge \frac{1}{2^{2n}a_k} \Rightarrow \frac{x^k_j}{a_j} >
\frac{x^k_k}{a_k} 
\]
The above conditions imply that $U(\xx)=\min_j \frac{x^k_j}{a_j}=\frac{x^k_k}{a_k} \Rightarrow a_k = \frac{x^k_k}{U(\xx^k)}$. Thus, $2d$ queries are enough
to learn $U$.

\begin{theorem}\label{thm.vq}
We can learn
\begin{itemize}
\item $\CH_{lin}$ from $O(d)$
\item $\CH_{slpc}$ from $O(n\kappa d)$ (where $\kappa =\max_j |U_j|$)
\item $\CH_{ces}^{\rho}$ from $O(d)$
\item $\CH_{leon}$ from $O(d)$
\end{itemize}
value queries.
\end{theorem}

\end{document}